\documentclass[letterpaper, 10 pt, conference]{ieeeconf}  

\IEEEoverridecommandlockouts                              

\usepackage{graphicx}      
\usepackage{amsmath,amssymb,amsfonts}
\DeclareMathOperator{\diag}{diag}
\usepackage{soul}
\usepackage{comment}
\usepackage[utf8]{inputenc} 
\usepackage{tikz}
\usepackage{algorithmic}
\usepackage{graphicx}
\usepackage{textcomp}
\usepackage{xcolor}
\usepackage{subfigure}
\usepackage{cite}
\usepackage{paralist}
\usepackage{graphicx}
\usepackage{figsize}
\usepackage{hyperref}

\newtheorem{problem}{Problem}[section]
\newtheorem{proposition}{Proposition}[section]

\newtheorem{remark}{Remark}[section]

\newtheorem{theorem}{Theorem}[section]

\usepackage{dsfont}

\usepackage{caption}
\title{\LARGE \bf Quaternion-based attitude stabilization via discrete-time IDA-PBC}

\author{M. Mattioni$^1$, A. Moreschini$^1$, S. Monaco$^1$ and D. Normand-Cyrot$^2$
\thanks{Supported by \emph{Sapienza Universit\`a di Roma} (\emph{Progetti di Ateneo 2021}-Piccoli progetti \emph{RP12117A7B268663}).}
\thanks{$^1$Dipartimento di Ingegneria Informatica, Automatica e Gestionale \emph{A. Ruberti} (Università degli Studi di Roma \emph{La Sapienza}); Via Ariosto 25, 00185 Rome, Italy {\tt\small {\{mattia.mattioni, salvatore.monaco\}@uniroma1.it}}.
}%
\thanks{$^2$Laboratoire de Signaux et Syst\`emes (L2S, CNRS, Sup\' elec and \emph{Univ. Paris-Saclay}); 3, Rue Joliot Curie, 91192, Gif-sur-Yvette, France {\tt\small  dorothee.normand-cyrot@centralesupelec.fr}}%
}

\begin{document}

\maketitle
\thispagestyle{empty}
\pagestyle{empty}

\begin{abstract} 
In this paper, we propose a new sampled-data controller for stabilization of the attitude dynamics at a desired constant configuration. The design is based on discrete-time interconnection and damping assignment (IDA) passivity-based control (PBC) and the recently proposed Hamiltonian representation of discrete-time nonlinear dynamics. Approximate solutions are provided with simulations illustrating performances.
\end{abstract}


\section{Introduction}

Attitude stabilization aims at orienting a physical object in relation with a specified inertial frame of reference and thus it applies in several engineering contexts such as aerospace or aerial robotics. 
Over the past twenty years numerous continuous-time quaternion-based attitude control laws have been proposed based on several design methodologies such as feedback-linearization \cite{kim2007robust,navabi2017spacecraft}, backstepping \cite{mino2019backstepping}, sliding mode \cite{zhu2011adaptive}, optimal control or learning \cite{wen1991attitude, hassrizal2016survey}. Among these, the ones relying upon passivity-based control (PBC) are outstanding as providing, in general, simple linear control laws \cite{lizarralde1996attitude, tsiotras1998further, di2003passive} in the form of P(I)D \cite{qasim2017pid}. 
Such methods involve different parametrizations of the attitude \cite{hughes2012spacecraft} as for instance Euler angles, the Euler-Rodrigues parameters, axis-angle, etc. Among these, the ones involving quaternions are of interest and widely used motivated by the lack of singularities in general and the fact that they are computationally less intense  \cite{wen1991attitude,kim2007robust}. However, even if quaternions can represent all possible attitudes, this representation is not unique: each attitude corresponds to two different quaternions. This typically gives rise to undesirable phenomena such as unwinding \cite{bhat2000topological, chaturvedi2011rigid}.

Very few results are available for addressing the design under sampling, that is when measurements are sampled and the input is piecewise constant  \cite{monaco2007advanced}. A first solution was proposed in \cite{monaco1986linearizing} based on sampled-data multi-rate feedback linearization. Such a solution requires a preliminary continuous-time control making the model finitely discretizable. Accordingly, if on one side such a control ensures one step convergence to the desired attitude, it requires a notable control effort and lacks in robustness with respect to unmodelled dynamics. In  \cite{mattioni2017immersion}, based on the Euler-Rodrigues kinematic model, a different  multi-rate digital control is designed in the context of Immersion and Invariance while relaxing the need of a preliminary continuous-time feedback.
In \cite{jiang2017sampled}, a single-rate sampled-data controller is proposed when solving a LQR problem on the approximate linear model at the desired attitude using quaternions. Such a feedback ensures local stabilization of the closed loop provided a set of LMIs is solvable off-line for a fixed value of the model parameters, the desired attitude and the sampling period. 

The contribution of this work stands in providing a new scalable digital control law involving single-rate sampling and quaternion description of the kinematics. The solution we propose is based on discrete-time Interconnection and Damping Assignment (IDA)-PBC\cite{moreschini2020stabilization} over the sampled-data equivalent model associated to the attitude dynamics with the aim of assigning a suitably defined discrete port-controlled Hamiltonian (pcH) structure \cite{moreschini2019discrete, 9749871}. Accordingly, the control is constructively proved to be the solution to a discrete matching equality. Even if exact closed-form solutions to this equality are tough to compute in practice, a recursive algorithm is proposed for computing approximate controllers at all desired orders.  We underline that, as a byproduct and minor contribution, we also provide a new continuous-time PBC controller for attitude stabilization.

\smallskip 
The rest of the paper is organized as follows. In Section \ref{sec:ct_st}, the attitude control problem is recalled and solved via continuous-time PBC. In Section \ref{sec:pb_sd}, the problem is formally set in a digital context over the sampled-data equivalent model of the attitude equation. The main result is in Section \ref{sec:main_sect} with comparative simulations in Section \ref{sec:sim}. Section \ref{sec:conclusions} concludes the paper with future perspectives.
\smallskip

\subsubsection*{Notations} 
Given a differentiable function $V:\mathbb{R}^n\to\mathbb{R}$, $\nabla V(\cdot)$ represents the gradient column-vector with $\nabla= \text{col}\{ \frac{\partial}{\partial x_i}\}_{i={1, \dots, n}}$ and $\nabla^2 V(\cdot)$ its Hessian. 
For $v, w \in \mathbb{R}^n$, the discrete gradient $\bar{\nabla}V|_{v}^{w}: \mathbb{R}^n \times \mathbb{R}^n \to \mathbb{R}^n$ is defined as  
\begin{align*} \bar{\nabla} V|_{v}^{w}= \int_0^1 \!\nabla V(v+s(w-v)) \ \textrm{d} s
\end{align*}
satisfying
$V(w)-V(v)=	(w-v)^\top\bar{\nabla}V|_{v}^{w}$
with $\bar{\nabla}V|_{v}^{v}={\nabla}V(v)$.
	 $I_n$ (or $I$ when clear from the context) and $I_d$ denote respectively the identity matrix of dimension $n$ and identity operator. $\mathbf{0}$ denotes the zero matrix of suitable dimensions. 
	Given a smooth (i.e., infinitely differentiable) vector field $f(\cdot)$ over $\mathbb{R}^n$,  $\mathrm{L}_{f}=\sum_{i=1}^{n}f_i(x)\frac{\partial}{\partial x_i}$ is the Lie operator with, recursively, $\mathrm{L}^i_f = \mathrm{L}_f \mathrm{L}_f^{i-1}$ and $\mathrm{L}_f^0 = I_d$. The exponential Lie series operator is defined as
	  $e^{\mathrm{L}_f}=I_d + \sum_{i\geq 1}\frac{1}{i!}\mathrm{L}_{f}^i.$  A function $R(x,\delta): \mathcal{B} \times \mathbb{R}\to \mathbb{R}^n$ is said in $O(\delta^p)$, with $p \geq 1$, if it can be written as $R(x, \delta) = \delta^{p-1}\tilde R(x, \delta)$ for all $x \in \mathcal{B}$ and there exist a function $\theta \in \mathcal{K}_{\infty}$ and $\delta^* >0$ s.t. $\forall \delta\leq \delta^*$, $| \tilde R (x, \delta)| \leq \theta(\delta)$. The symbols $>$ and $<$ denote  positive and negative definite functions whereas $\prec$ and $\succ$ ($\preceq$ and $\succeq$) positive and negative (semi) definite matrices.
	  Given a vector $\omega = \text{col}\{\omega_x, \omega_y, \omega_z \} \in \mathbb{R}^3$, we denote 
	  \begin{align*}
	      S(\omega) = \begin{pmatrix}
                        0 & -\omega_z & \omega_y \\
                        \omega_z & 0 & - \omega_x\\
                        -\omega_y & \omega_x & 0
	                \end{pmatrix} = -S^\top (\omega).
	  \end{align*}
	 Given two matrices $A \in \mathbb{R}^{n_1\times n_2}$ and $B \in \mathbb{R}^{m_1\times m_2}$, $
A \otimes B $ denotes the Kronecker product. 
Given a matrix $A  \in\mathbb{R}^{n\times n}$, with elements $a_{ij}\in \mathbb{R}$ (for $i,j = 1, \dots, n$), $\text{vec}(A) := \text{col}\{ 
a_{11}, \dots,  a_{1n}, 
\dots, 
a_{n1},  \dots,  a_{nn}
\}\in \mathbb{R}^{n^2}$ denotes the vectorization so that
$
Ab = (I \otimes b^\top) \text{vec}(A)
$
for $b \in \mathbb{R}^n$.

\section{Quaternion-based attitude equations and control: the continuous-time case}
\label{sec:ct_st}
In the following, let the kinematics and dynamics of a rigid body orientation in the inertial coordinates be described by
\begin{subequations}\label{eq:quat_d}
\begin{align}
    \dot q_0 =& \frac{1}{2}S(\omega) q_0 + \frac{1}{2}\omega q_r \label{eq:quat_d_a}, \quad 
    \dot q_r = - \frac{1}{2}\omega^\top q_0 \\
    M \dot \omega =& S(\omega) M \omega + u \label{eq:quat_d_c}
\end{align}
\end{subequations}
where $q = \text{col}\{q_0,\ q_r \}\in \mathbb{R}^3 \times \mathbb{R}$ is the quaternion vector (also known as Euler parameters) verifying
\begin{align}
    q^\top q = q_0^\top q_0 + q_r^2 = 1, 
\end{align}
$\omega \in \mathbb{R}^3$ is the angular velocity,  $u\in \mathbb{R}^3$ the input torque and  $M = M^\top \succ 0$ the inertia matrix.

Let $q^\star = \text{col}\{ q_{0}^{\star}, q_{r}^{\star}\}$ be a desired attitude configuration. Then, the attitude control problem above can be addressed over the error quaternion defined as $\varepsilon = \text{col}\{\varepsilon_0, \varepsilon_r \}$ \cite{hughes2012spacecraft}
\begin{align*}
    \varepsilon_0 = & q_r q_{0}^{\star} +  q_{r}^{\star} q_0 + S(q_0)q_{0}^{\star}, \quad \varepsilon_r =  q_r q_r^\star - q_0^\top q_0^\star
\end{align*}
with the corresponding dynamics reading
\begin{subequations}\label{eq:quat_d_err}
\begin{align}
    \dot \varepsilon_0 =& \frac{1}{2}S(\omega) \varepsilon_0 + \frac{1}{2}\omega \varepsilon_r \label{eq:quat_d_err_a}, \quad 
    \dot \varepsilon_r = - \frac{1}{2}\omega^\top \varepsilon_0 \\
    M \dot \omega =& S(\omega) M \omega + u. \label{eq:quat_d_err_c}
\end{align}
\end{subequations}
Accordingly, asymptotically stabilizing $q^\star$ for \eqref{eq:quat_d}  corresponds to asymptotically stabilizing \eqref{eq:quat_d_err} at the equilibrium
\begin{align}\label{eq:e_star}
    \varepsilon^\star = \begin{pmatrix}
    \mathbf{0}^\top & 
    1
    \end{pmatrix}^\top, \quad \omega^\star = \mathbf{0}.
\end{align}

Among the solutions available in continuous time, the ones relying upon PBC yield a family of linear control laws \cite{lizarralde1996attitude, tsiotras1998further} of the form
\begin{align}\label{eq:lin_fed_ct}
    u = -\kappa_{0} \varepsilon_0 - \kappa_\text{di}\omega, \quad \kappa_0, \kappa_\text{di} \succ 0.
\end{align}
In the following, we endorse such controllers with an IDA-PBC specification \cite{ortega2002interconnection}. 
For, we note that \eqref{eq:quat_d_err} exhibits the conservative pcH form 
\begin{align*}
    \begin{pmatrix}
    \dot \varepsilon_0 \\
    \dot \varepsilon_r \\
    M \dot \omega
    \end{pmatrix} = \begin{pmatrix}
    \frac{1}{2}S(\omega) & \frac{1}{2}\omega & \mathbf{0}  \\
    -\frac{1}{2}\omega^\top & 0 & \mathbf{0} \\
    \mathbf{0} & \mathbf{0} & S(\omega)
    \end{pmatrix}\begin{pmatrix}
    \varepsilon_0 \\
    \varepsilon_r \\
    M  \omega
    \end{pmatrix} + \begin{pmatrix}
    \mathbf{0}\\
    0 \\
    I
    \end{pmatrix}u
\end{align*}
with quadratic Hamiltonian and interconnection matrix
\begin{align*}
    H(\varepsilon, \omega)\!\! =&\!\! \frac{1}{2}\big(\varepsilon^\top \varepsilon + \omega^\top M \omega)
,\ 
    J(\omega)\!\!\! =\!\!\! \begin{pmatrix}
    \frac{1}{2}S(\omega) & \frac{1}{2}\omega & \mathbf{0}  \\
    -\frac{1}{2}\omega^\top & 0 & \mathbf{0} \\
    \mathbf{0} & \mathbf{0} & S(\omega)
    \end{pmatrix}.
\end{align*}
On this basis, the following result holds. 
\begin{proposition}
\label{prop:ct_ida}
Consider the error dynamics \eqref{eq:quat_d_err} and the linear feedback \eqref{eq:lin_fed_ct} making \eqref{eq:e_star} asymptotically stable. 
When $
    \kappa_0 =\frac{1}{2}M^{-1},
$
\eqref{eq:lin_fed_ct} is an IDA-PBC feedback of the form
\begin{equation}\label{eq:uida_ct}
\begin{split}
    & u_\text{ida}(\varepsilon, \omega) = u_\text{es}(\varepsilon) + u_\text{di}(\omega)\\
    & u_\text{es}(\varepsilon) = - \frac{1}{2}M^{-1}\varepsilon_0, \quad u_\text{di}(\omega) = -\kappa_\text{di} \omega
    \end{split}
\end{equation}
assigning the closed-loop pcH form
\begin{align*}
    \begin{pmatrix}
    \dot \varepsilon_0 \\
    \dot \varepsilon_r \\
    M \dot \omega
    \end{pmatrix} = \begin{pmatrix}
    \frac{1}{2}S(\omega) & \frac{1}{2}\omega & \frac{1}{2}M^{-1}  \\
    -\frac{1}{2}\omega^\top & 0 & \mathbf{0} \\
   -\frac{1}{2}M^{-1} & \mathbf{0} & S(\omega) - \kappa_\text{di} M^{-1}
    \end{pmatrix}\begin{pmatrix}
     \varepsilon_0 \\
     \varepsilon_r -1\\
    M  \omega
    \end{pmatrix}
\end{align*}
with 
\begin{subequations}
\begin{align}
    & H_d(\varepsilon, \omega) = \frac{1}{2} \Big(
    \varepsilon_0^\top \varepsilon_0 + (\varepsilon_r-1)^2 + \omega^\top M \omega\Big)\label{eq:Hd}\\
    & J_d(\omega) = \begin{pmatrix}
    \frac{1}{2}S(\omega) & \frac{1}{2}\omega & \frac{1}{2}M^{-1}  \\
    -\frac{1}{2}\omega^\top & 0 & \mathbf{0} \\
   -\frac{1}{2}M^{-1} & \mathbf{0} & S(\omega)
    \end{pmatrix} \label{eq:Jd}
    \\ &
    R_d = \begin{pmatrix}
    \mathbf{0} & \mathbf{0} & \mathbf{0}  \\
    \mathbf{0} & 0 & \mathbf{0} \\
   \mathbf{0} & \mathbf{0} & \kappa_\text{di} M^{-1}
    \end{pmatrix}\succeq 0.\label{eq:Rd}
\end{align}
\end{subequations}
\end{proposition}

Before going farther we note that, as deeply commented in \cite{bhat2000topological, mazenc2016quaternion}, regulation to the desired orientation is, in general, achieved when $\varepsilon(t)\to \pm \varepsilon^\star$. However, in the result above (and throughout the whole paper), we fix the equilibrium to stabilize as the positive one while neglecting possible unwinding phenomena whose study is left as a perspective. With this in mind and because $-\varepsilon^\star$ is a non-isolated equilibrium of the closed-loop dynamics, all upcoming results are intended to hold locally, unless explicitly specified. 

\begin{remark} \label{rmk:new_un}
To cope with unwinding, the energy-shaping controller in \eqref{eq:uida_ct} can be modified to assign the new energy
\begin{align*}
    & \hat H_d(\varepsilon, \omega) = \frac{1}{2} \Big(
    \varepsilon_0^\top \varepsilon_0 +1-\varepsilon_r^2 + \omega^\top M \omega\Big)
\end{align*}
with local minima at $\pm \varepsilon^\star$. This is achieved replacing in \eqref{eq:uida_ct} 
\begin{equation*}
\begin{split}
    & \hat u_\text{es}(\varepsilon) = - \varepsilon_r M^{-1}\varepsilon_0
    \end{split}
\end{equation*}
assigning the closed-loop pcH form
\begin{align*}
    \begin{pmatrix}
    \dot \varepsilon_0 \\
    \dot \varepsilon_r \\
    M \dot \omega
    \end{pmatrix}\! \! \! = \! \! \!\begin{pmatrix}
    \frac{1}{2}S(\omega) & 0 & \frac{1}{2}\varepsilon_r M^{-1}  \\
    0 & 0 & -\frac{1}{2}\varepsilon_{0}^\top M^{-1} \\
  - \frac{1}{2}\varepsilon_r M^{-1} & \frac{1}{2} M^{-1}\varepsilon_0 & S(\omega) - \kappa_\text{di} M^{-1}
    \end{pmatrix}\! \! \! \begin{pmatrix}
     \varepsilon_0 \\
     -\varepsilon_r\\
    M  \omega
    \end{pmatrix}.
\end{align*}
However, such a feedback guarantees asymptotic stabilization of the desired attitude (associated to equilibria $\pm \varepsilon^\star$) provided that $\varepsilon_r(0)\neq 0$ or $\omega(0)\neq 0$. How to overcome such a limitation is currently under investigation.
\end{remark}

In the following, starting from Proposition \ref{prop:ct_ida}, we design a piecewise constant control law driving the body to the desired orientation associated to  $q^\star = \text{col}\{ q_{0}^{\star}, q_{r}^{\star}\}$. 

\section{Sampled-data model and problem statement} \label{sec:pb_sd}

Consider the system \eqref{eq:quat_d_err} and let measurements (of the state) be available at the sampling instants $t = k\delta$, with $k\geq 0$ and $\delta>0$ the sampling period, and the input be piecewise constant over time intervals of length $\delta$; i.e. we set
$
    u(t) = u_k$, $t\in [k\delta, (k+1)\delta[.
$
Then, we seek for a sampled-data feedback $u_k = \gamma(\varepsilon_k, \omega_k)$ making \eqref{eq:e_star} asymptotically stable for the sampled-data equivalent model of \eqref{eq:quat_d_err} given by 
\begin{align}\label{eq:sd_model}
    \zeta^{+}(u)
     =& 
     \zeta + \delta F^\delta(\zeta, u)
\end{align}
with $\zeta = \text{col}\{\varepsilon, \omega\}$, $\zeta = \zeta_k$, $u = u_k$, $\zeta^+(u) = \zeta_{k+1}$ and 
\begin{align*}
 f(\zeta) =& f(\varepsilon, \omega) = J(\omega)\nabla H(\varepsilon,\omega), \quad B = \begin{pmatrix}
    \mathbf{0}^\top &
    I
\end{pmatrix}^\top 
\\ 
   F^\delta(\zeta, u) =&  F^\delta(\varepsilon, \omega, u) =  J(\omega)\nabla H(\varepsilon,\omega)  + Bu\\ &+ \sum_{i>0}\frac{\delta^i}{(i+1)!} \mathrm{L}_{f  + Bu}^{i} \Big(J(\omega)\nabla H(\varepsilon,\omega) +Bu\Big).
\end{align*}
Setting $F^\delta_0(\zeta) = F^\delta(\zeta,\mathbf{0})$, we denote by
\begin{align*}
\zeta^+
     :=&  
     \zeta^{+}(\mathbf{0}) =  \zeta + \delta F^\delta_0(\zeta)\\
    g^\delta (\varepsilon, \omega, u) u : =& F^\delta(\zeta, u) - F^\delta_0(\zeta)
\end{align*}
the drift and controlled components associated to \eqref{eq:sd_model}.
In this setting we wish to accomplish stabilization of $\zeta^\star = \text{col}\{ \varepsilon^\star, \mathbf{0}
\}$ as in \eqref{eq:e_star} by solving a discrete-time IDA-PBC problem over the sampled-data equivalent model \eqref{eq:sd_model} in two steps, as formalized in \cite{moreschini2020stabilization} and recalled below. 
\begin{problem}[sampled-data energy-shaping]
\label{pb:es_sd}
Design the energy-shaping control $u^\delta_\text{es}: \mathbb{R}^{4}\times \mathbb{R}^3 \to \mathbb{R}^3$ 
so that, setting
\begin{align}\label{eq:u_partial}
    u = u^\delta_\text{es}(\zeta) + v
\end{align} 
the closed-loop dynamics exhibits a conservative discrete-time pcH structure \cite{moreschini2019discrete, 9749871}
\begin{align}
    \label{eq:dtpcH_cl}
\zeta^{+}(u^\delta_\text{es}(\zeta)+v)\!\! =\!\! \zeta\!\!+\!\! \delta J^\delta_d(\zeta )\bar \nabla H_d|_{\zeta}^{\zeta^{+}(u^\delta_\text{es}(\zeta))}\!\! +\!\! \delta g^\delta_d(\zeta, v) v
\end{align}
with the new Hamiltonian \eqref{eq:Hd}, discrete gradient
\begin{align*}
    \bar \nabla H_d|_{\zeta}^{\zeta^{+}} = \frac{1}{2}P (\zeta^+ + \zeta - 2 \zeta^\star), \quad P = \begin{pmatrix}
    I & \mathbf{0} 
    \\
    \mathbf{0} & M
    \end{pmatrix}\succ 0
\end{align*}
a suitable $J^\delta_d(\zeta) = -\big( J^\delta_d(\zeta) \big)^\top \in \mathbb{R}^{7\times 7}$ and 
 \begin{align*}
     g^\delta_d(\zeta, v)v=& g^{\delta}(\zeta, u^{\delta}_\text{es}(\zeta)+v)(u^{\delta}_\text{es}(\zeta)+v)\\ &-
     g^{\delta}(\zeta, u^{\delta}_\text{es}(\zeta)) u^{\delta}_\text{es}(\zeta). \hfill \blacksquare
 \end{align*}
 \end{problem}
 
 \smallskip
The solution to Problem \ref{pb:es_sd} guarantees that the closed-loop dynamics \eqref{eq:dtpcH_cl} possesses a stable equilibrium at \eqref{eq:e_star}. In addition, it is proved that \eqref{eq:u_partial} makes \eqref{eq:dtpcH_cl} lossless \cite{monaco2011nonlinear} with respect to the conjugate output
\begin{align}
    \label{eq:pass_con_sd}
    Y^\delta_d(\zeta,v) = \big(g_d^\delta(\zeta,v)\big)^\top \bar \nabla H|_{\zeta^{+}(u^\delta_\text{es}(\zeta))}^{\zeta^{+}(u^\delta_\text{es}(\zeta)+v)} 
\end{align}  
that is, it verifies the dissipation equality 
\begin{align}\label{eq:EBE}
    \Delta H_d(\zeta) \!\!= \!\! H_d(\zeta^+(u^\delta_\text{es}(\zeta)+v)) \!\!-\!\! H_d(\zeta) \!\!=\!\!\delta  v^\top Y^\delta_d(\zeta,v).
\end{align}
With this in mind, attitude stabilization via IDA-PBC is achieved under damping if the problem below is solved.
\begin{problem}[sampled-data damping injection]
\label{pb:di_sd} Seek, if any, for a damping-injection feedback $v = u^\delta_\text{di}(\zeta)$ 
with $u^\delta_\text{di}: \mathbb{R}^4 \times \mathbb{R}^3\to \mathbb{R}^3$ solution to the damping equality
\begin{align}\label{eq:u_di}
    u^\delta_\text{di}(\zeta) + \kappa_\text{di} Y^\delta_d(\zeta,u_\text{di}^\delta(\zeta)) = \mathbf{0},\quad  \kappa_\text{di}\succ 0
\end{align}
so making $\zeta^\star$ asymptotically stable for  \eqref{eq:dtpcH_cl}. $\hfill \blacksquare$
\end{problem}

\section{Digital attitude stabilization via IDA-PBC}\label{sec:main_sect}
Now proceed with the design in two steps by proving, in a constructive manner, the existence of both the energy-shaping and damping-injection components of the control
\begin{align}\label{eq:uida_sd}
    u = u^\delta_\text{ida}(\zeta) = u^\delta_\text{es}(\zeta) + u^\delta_\text{di}(\zeta)
\end{align}
aimed at, respectively, assigning and asymptotically stabilizing the desired equilibrium with a desired energy via PBC.
\subsection{Main result}
As detailed in Problem \ref{pb:es_sd} the energy-shaping component is responsible for assigning the equilibrium \eqref{eq:e_star} with a discrete-time pcH form \eqref{eq:dtpcH_cl}. As proved in \cite{moreschini2019discrete} in a purely discrete-time context, this corresponds to solving the so-called Discrete-time Matching Equation (DME) below
\begin{align}
    \label{eq:DME}
       J^\delta_d(\zeta)P\Big( \zeta -\zeta^\star  + \frac{\delta}{2} F^\delta(\zeta, u^\delta_\text{es}(\zeta)) \Big) = F^\delta(\zeta, u^\delta_\text{es}(\zeta))
\end{align}
for a suitably defined skew-symmetric interconnection matrix $J^\delta_d(\zeta)\in \mathbb{R}^{7\times 7}$. The next result proves that a solution to \eqref{eq:DME} exists in the form of series expansions in powers of $\delta$ around the continuous-time counterparts in Proposition \ref{prop:ct_ida}. 

\smallskip
\begin{proposition}\label{prop:sd_es}
Consider the error dynamics \eqref{eq:quat_d_err} with sampled-data equivalent model \eqref{eq:sd_model} and $\zeta^\star =\text{col}\{ \varepsilon^\star, \mathbf{0}\}$ as in \eqref{eq:e_star} the equilibrium to stabilize. Then, there exists $T^\star >0$ such that for all $\delta \in [0, T^\star[$ the DME \eqref{eq:DME} admits unique solutions $<J^\delta_d(\zeta), \ u_\text{es}^\delta(\zeta)>$ in the form of series expansions around \eqref{eq:uida_ct}-\eqref{eq:Jd}; namely, one gets
\begin{subequations}
\label{eq:JdUes}
\begin{align}
    J_d^\delta(\zeta) =& J_d(\omega) + \sum_{i>0}\frac{\delta^i}{(i+1)!}J_d^i(\zeta)
    \label{eq:JdUes_a}
    \\
    u_\text{es}^\delta(\zeta) =& u_\text{es}(\varepsilon) + \sum_{i>0}\frac{\delta^i}{(i+1)!}u_\text{es}^i(\zeta) \label{eq:JdUes_b}
\end{align}
\end{subequations}
with the superscripts $^i$ denoting the order of the term in the series expansion.
\end{proposition}
\begin{proof}
Before going through the details of the proof we highlight that, because $J_d^\delta(\zeta)$ must be skew-symmetric, then out of $49$ elements, only $28$ must be identified; accordingly, the equality \eqref{eq:DME} is in $28+3$ unknowns. With this in mind, the proof follows from the Implicit Function Theorem rewriting \eqref{eq:DME} as a formal series equality in powers of $\delta$ in the corresponding $31$ unknowns; namely, setting for simplicity $u = u_\text{es}^\delta(\zeta)$,  $\mathcal{J} = J_d^\delta(\zeta)$ and $j = \mathrm{vec}(\mathcal{J}) \in \mathbb{R}^{49}$, one looks for the solutions to the equality $\mathcal{Q}(\delta, \zeta, j, u) = \mathbf{0}$
 where
\begin{align*}
\mathcal{Q}(\delta, \zeta, j, u) := 
 F^\delta(\zeta,u) -  (I\otimes \bar \nabla^\top H_d|_\zeta^{\zeta^+(u)} )j 
 \end{align*}
admits the expansion 
\begin{align*}
\mathcal{Q}(\delta, \zeta, j, u) = \mathcal{Q}^0(\zeta,j, u) + \sum_{i >0} \frac{\delta^i}{(i+1)!}\mathcal{Q}^i(\zeta,j, u).
\end{align*}
for suitably defined terms $\mathcal{Q}^i(\zeta,j, u) $ of order $i\geq 0$.
Because
 \begin{align*}
\mathcal{Q}^0(\zeta,j, u) =& 
f(\zeta)+B u - \mathcal{J} \nabla H_d(\zeta)
\\ =&
f(\zeta)+B u  - (I \otimes \nabla^\top H_d(\zeta)) j , 
\end{align*}
one gets that the corresponding equality $\mathcal{Q}^0(\zeta,j, u) = \mathbf{0}$ is solved when $j = \text{vec}(J_d(\omega))$ (i.e., $\mathcal{J} = J_d(\omega)$) and $u = u_\text{es}(\varepsilon)$ as in \eqref{eq:uida_ct}-\eqref{eq:Jd}.
By the Implicit Function Theorem then, the formal equality $\mathcal{Q}(\delta, \zeta, j, u) = \mathbf{0}$ admits a solution of the form \eqref{eq:JdUes} because the matrix
\begin{align*}
\lim_{\delta \to 0} \begin{pmatrix}
\frac{\partial }{\partial u} & 
\frac{\partial }{\partial j} 
\end{pmatrix}\mathcal{Q}(\delta, \zeta, j, u) =
\begin{pmatrix}
B &  I \otimes \nabla^\top H_d(\zeta)
\end{pmatrix} \\ =
\begin{pmatrix}
\mathbf{0} & I_4 \otimes (\zeta -\zeta^\star)^\top P & \mathbf{0}  \\
I_3 & \mathbf{0} & I_3 \otimes (\zeta -\zeta^\star)^\top P
\end{pmatrix}
\end{align*}
is full rank at  $\zeta \neq \zeta^\star$ so getting the result. 
\end{proof}

By the result above, the feedback law \eqref{eq:u_partial} assigns the conservative discrete pcH structure \eqref{eq:dtpcH_cl} with interconnection matrix of the form \eqref{eq:JdUes_a} and the same Hamiltonian as in continuous time. As a consequence \cite{moreschini2019discrete}, the controlled system is passive (and lossless) with an equilibrium at the desired $\zeta^\star = \text\{\varepsilon^\star, \mathbf{0} \}$. Accordingly, if one is able to compute a solution to the damping equality \eqref{eq:u_di} asymptotic stabilization of the desired equilibrium is achieved. 

\begin{theorem}
\label{th:fin_sd}
Consider the dynamics \eqref{eq:quat_d_err} under the hypotheses of Proposition \ref{prop:sd_es} with \eqref{eq:u_partial} assigning the pcH form \eqref{eq:dtpcH_cl}. Then, the following holds:
\begin{description}
\item[$(i)$] the dynamics \eqref{eq:dtpcH_cl} is lossless with respect to the output \eqref{eq:pass_con_sd} and storage function \eqref{eq:Hd}, that is, the energy-balance equality \eqref{eq:EBE} holds;
\item[$(ii)$] there exists $T^\star >0 $ such that for all $\delta \in [0, T^\star[$ the damping equality \eqref{eq:u_di} admits a unique solution in the form of a series expansion around \eqref{eq:uida_ct}, that is
\begin{align}
    \label{eq:u_di_ser}
    u_\text{di}^\delta(\zeta) = u_\text{di}(\omega) + \sum_{i>0}\frac{\delta^i}{(i+1)!}u_\text{di}^i(\zeta)
\end{align}
with $^i$ denoting the order of the term in the series expansion; 
\item[$(iii)$] the IDA-PBC feedback \eqref{eq:uida_sd} asymptotically stabilizes \eqref{eq:dtpcH_cl} at the desired equilibrium \eqref{eq:e_star}.
\end{description}
\end{theorem}
\begin{proof}
$(i)$ follows by computing the one step increment of the Hamiltonian \eqref{eq:Hd} along \eqref{eq:dtpcH_cl}; namely under Proposition \ref{prop:sd_es} one gets 
\begin{align*}
   \Delta H_d   =& H_d(\zeta^+(u^\delta_\text{es}(\zeta)+v)) - H_d(\zeta^+(u^\delta_\text{es}(\zeta)))  \\ &
   +H_d(\zeta^+(u^\delta_\text{es}(\zeta))) - H_d(\zeta) \\ = &  H_d(\zeta^+(u^\delta_\text{es}(\zeta)+v)) - H_d(\zeta^+(u^\delta_\text{es}(\zeta))) 
\end{align*}
because $J_d^\delta(\zeta) = -\big( J_d^\delta(\zeta)\big)^\top$ and
\begin{align*}
   &  H_d(\zeta^+(u^\delta_\text{es}(\zeta))) - H_d(\zeta)  \\ & = \bar \nabla^\top H_d|_{\zeta}^{\zeta^+(u^\delta_\text{es}(\zeta))}J_d^\delta(\zeta) \bar \nabla  H_d|_{\zeta}^{\zeta^+(u^\delta_\text{es}(\zeta))}  = 0
\end{align*}
so that
\begin{align*}
    \Delta H_d  =\delta  \bar \nabla^\top H_d|_{\zeta^+(u^\delta_\text{es}(\zeta))}^{\zeta^+(u^\delta_\text{es}(\zeta)+v)} g^{\delta}_d(\zeta,v)v  
\end{align*}
and thus \eqref{eq:EBE}. 
$(ii)$ follows from the implicit function theorem along the lines of the proof of Proposition \ref{prop:sd_es}. As far as $(iii)$ is concerned, because \eqref{eq:u_di_ser} solves \eqref{eq:u_di}, it guarantees
\begin{align*}
    \Delta H_d  = -\delta \kappa_\text{di} \|Y^\delta(\zeta, u_\text{di}^\delta (\zeta))\|^2\leq 0, \quad \kappa_\text{di}\succ 0
\end{align*}
with trajectories of the closed-loop system converging to the largest invariant set contained in $\{\zeta \in \mathbb{R}^7 \text{ s.t. } Y^\delta(\zeta, u_\text{di}^\delta (\zeta)) = 0 \}$. By Proposition \ref{prop:ct_ida} and average passivity \cite{monaco2011nonlinear} such set only contains  $\zeta^\star$ and  $(iii)$ follows. 
\end{proof}

\begin{remark} 
As $\delta \to 0$, the overall attitude control law \eqref{eq:uida_sd} and the sampled-data interconnection matrix \eqref{eq:JdUes} recover the continuous-time counterparts in \eqref{eq:uida_ct} and \eqref{eq:Jd}.
\end{remark}
It turns out that the overall feedback law \eqref{eq:uida_sd} gets the form of an asymptotic series expansion in powers of $\delta$ through the energy-shaping and damping components \eqref{eq:JdUes_b} and \eqref{eq:u_di_ser}. Despite computing closed form solutions is a tough task, approximations can be naturally defined as detailed below.

\begin{remark}\label{rmk:Tstar}
In both Theorem \ref{th:fin_sd} and Proposition \ref{prop:sd_es} 
the existence of $T^\star$ depends on the state at the current time; namely, for each measured state at time $t = k\delta$ (i.e., $\zeta_k$), there exists $T^\star$ so that the corresponding equalities admit a solution. Current work is toward the definition of additional uniformity arguments for the qualitative definition of a single $T^\star$ for all $\zeta$ in a neighborhood of $\zeta^\star$.
\end{remark}

\subsection{Computational aspects}
The result in Proposition \ref{prop:sd_es} highlights that a solution to the matching equality \eqref{eq:DME} (and thus to Problem \ref{pb:es_sd}) exists in the form of a series expansion in powers of $\delta$. Despite closed forms are hard to be computed in practice, all terms $<J^i_d(\zeta), u_\text{es}^i(\zeta)>$ in \eqref{eq:JdUes} are explicitly computable as the solution to a set of linear equalities of the form
\begin{align}\label{eq:elli}
    B u^i_\text{es}(\zeta) - J^i_d(\zeta)\big(\zeta - \zeta^\star) = \ell^i(\zeta), \quad i = 0,1,\dots
\end{align}
with $\ell^i(\zeta) = \tilde \ell^i(\zeta, u_\text{es}(\zeta), J_d(\zeta), \dots, u_\text{es}^{i-1}(\zeta), J_d^{i-1}(\zeta))$ deduced by substituting \eqref{eq:JdUes} into \eqref{eq:DME} and equating the terms with the same powers of $\delta$. Setting blockwise
\begin{align}
    J_d^{i}(\zeta) = -\big(J_d^{i}(\zeta)\big)^\top
= \begin{pmatrix}
J_{11}^i(\zeta) & J_{12}^i(\zeta) & J_{13}^i(\zeta)\\
\star & 0 & J_{23}^i(\zeta)\\
\star & \star & J_{33}^i(\zeta) 
\end{pmatrix}
\label{eq:Jdi}
\end{align}
the solution can be iteratively computed by multiplying both sides of  \eqref{eq:elli} by $B^\perp = \begin{pmatrix} I_4 & \mathbf{0} \end{pmatrix}$ so getting
\begin{align*}
   & B^\perp J^i_d(\zeta)(\zeta - \zeta^\star) = B^\perp \ell^i(\zeta)\\
   & u^i_\text{es}(\zeta) =  B^\top \big(\ell^i (\zeta)+ J^i_d(\zeta)(\zeta - \zeta^\star)\big)
\end{align*}
 and, in particular, the linear equality in $J^i_d(\zeta)$ given by
\begin{align*}
   \begin{pmatrix}
J_{11}^i(\zeta) & J_{12}^i(\zeta) & J_{13}^i(\zeta)\\
-\big(J_{12}^i(\zeta)\big)^\top & 0 & J_{23}^i(\zeta)
\end{pmatrix}\begin{pmatrix}
\varepsilon_0\\
\varepsilon_r-1\\
M\omega 
\end{pmatrix} = B^\perp \ell^i(\zeta).
\end{align*}
Accordingly, for the first terms, one gets
\begin{align*}
    u^1_\text{es}(\zeta) =& -\frac{1}{4}M^{-1}\big(S(\omega)q_0 + q_r \omega\big)\\
    u^2_\text{es}(\zeta) =& \frac{3}{8}\big( \frac{1}{2}S(\dot \omega_0) -S^{2}(\omega) \big)M^{-1}q_0  + \frac{1}{2} \Big(\frac{1}{8}\big(3 I\\ & -M^{-1}\big)\omega \omega^\top - M^{-1}S(\dot \omega_0) - \frac{3}{8}M^{-1}S^2(\omega) \Big)q_0 \\ &- \frac{3}{8}M^{-1}\dot \omega_0 q_r
    \end{align*}
     $J^1_d(\zeta) = J_d(\dot \omega_0)$ and, using \eqref{eq:Jdi} for $i = 2$$,    J^{2}_{23}(\zeta) = \mathbf{0}$
    \begin{align*}
    J^2_{11}(\zeta) =&\frac{1}{2} S(\ddot \omega_0) + \frac{1}{8} M^{-1}S(\omega) M^{-1}+ \frac{3}{4}S(M^{-1}u^1_\text{es}(\zeta))\\ &
    -\frac{1}{6}S^{3}(\omega) +\frac{1}{2}\Big( 
    S(\omega) S(\dot \omega_0) - S(\dot \omega_0) S(\omega) \\ &+ \omega \dot \omega_0^\top - \dot \omega_0^\top \omega 
    \Big)\\
    J^2_{12}(\zeta) =& \frac{3}{4}M^{-1}u^1_\text{es}(\zeta ) + \frac{1}{16}\omega \omega^\top \omega + \frac{1}{2}\ddot \omega_0
    \\
    J^{2}_{13}(\zeta) =& \frac{1}{16}\omega \omega^\top M^{-1}+ \frac{1}{2}M^{-1}
    S^{2}(\omega) \\ & +\frac{1}{8}\Big( M^{-1}S(\dot \omega_0) - S(\dot \omega_0) M^{-1}
    \Big)\\
    J^{2}_{33}(\zeta) =& S(\ddot \omega_0) + \frac{3}{2}S(M^{-1}u^1_\text{es}(\zeta) -\frac{1}{2}S^3(\omega) \\ &+\frac{1}{2}
    \Big(S(\dot \omega_0) S(\omega) - S(\omega)S(\dot \omega_0) \Big)
\\
     \dot \omega_0 =& M^{-1} S(\omega)M\omega - \frac{1}{2}M^{-2}q_0\\
     \ddot \omega_0 =& M^{-1}S(\dot \omega_0) M \omega + M^{-1}S(\omega) \big(
     S(\omega)M\omega \!-\! \frac{1}{2}M^{-1}\varepsilon_0
     \big).
\end{align*}

Along the same lines, exact forms for the damping control in Theorem \ref{th:fin_sd} cannot be computed in practice. However, all terms of the damping injection component in \eqref{eq:u_di_ser} can be easily computed through an iterative procedure solving a linear equality in the corresponding unknown. Such an equality is deduced substituting \eqref{eq:u_di_ser} into \eqref{eq:u_di} so getting, for the first terms
\begin{align*}
    u_\text{di}^1(\zeta)   &= - \kappa_\text{di} M^{-1}\big( S(\omega) M - \kappa_\text{di}\big)\omega + \frac{\kappa_\text{di}}{2}M^{-2}\varepsilon_0\\
    u_\text{di}^2(\zeta) &= \dot u_\text{di}^1(\zeta)  - \kappa_\text{di} M\big( 
     u_\text{es}^1(\zeta) + \frac{1}{2} u_\text{di}^1(\zeta)
    \big)
\end{align*}
with, recalling the continuous-time control \eqref{eq:uida_ct},
\begin{align*}
   \dot u_\text{di}^1(\zeta) =& \frac{\partial u_\text{di}^1(\zeta) }{\partial \zeta }\Big( f(\zeta) + Bu_\text{ida}(\zeta) \Big).
\end{align*}
With this in mind, only IDA-PBC control laws defined as truncations, at all desired finite order $p\geq0$, of the corresponding series expansions \eqref{eq:JdUes_b} and \eqref{eq:u_di_ser} can be implemented in practice; namely, the $p^\text{th}$-order approximate IDA-PBC feedback is defined as
\begin{align} \label{eq:app_mech}
    u^{\delta, [p]}_\text{ida}(\zeta) = \sum_{\ell = 0}^{p} \frac{\delta^\ell}{(\ell+1)!}(u^\ell_\text{es}(\zeta) + u^\ell_\text{di}(\zeta)), p \geq 0.
\end{align}
For $p = 0$ \eqref{eq:app_mech} one recovers the  so-called emulation of \eqref{eq:uida_ct}, that is the continuous-time control implemented via sample-and-hold devices and no redesign taking into account the effect of sampling (see e.g., \cite{mazenc2013robustness}).
Such controllers ensure practical asymptotic stability of the desired equilibrium in closed loop \cite[~Proposition~4.2]{mattioni2017immersion}; i.e., trajectories converge to a neighborhood of the desired attitude with radius $\delta^{i+1}$.

\begin{remark}\label{rmk:meas_const}
The terms above highlight that, contrarily to the continuous-time case, the sampled-data feedback law explicitly depends on the sign of the quaternion which must be then consistently reconstructed from measurements.
\end{remark}

\begin{figure*}
\centering
\SetFigLayout{2}{2}
  \subfigure[Continuous-time IDA-PBC]{\includegraphics{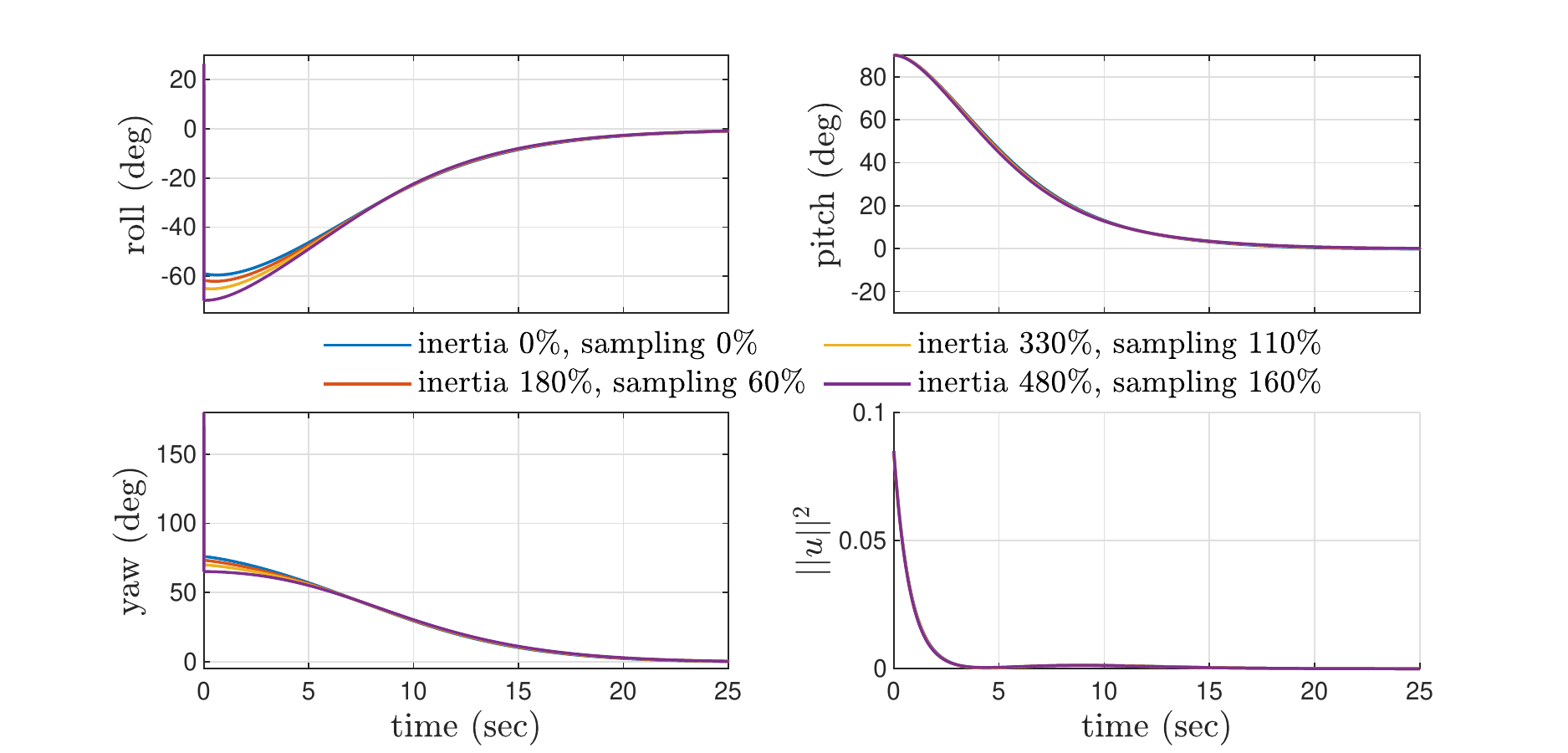}}\hspace{-2em}
  \subfigure[$p=0$]{\includegraphics{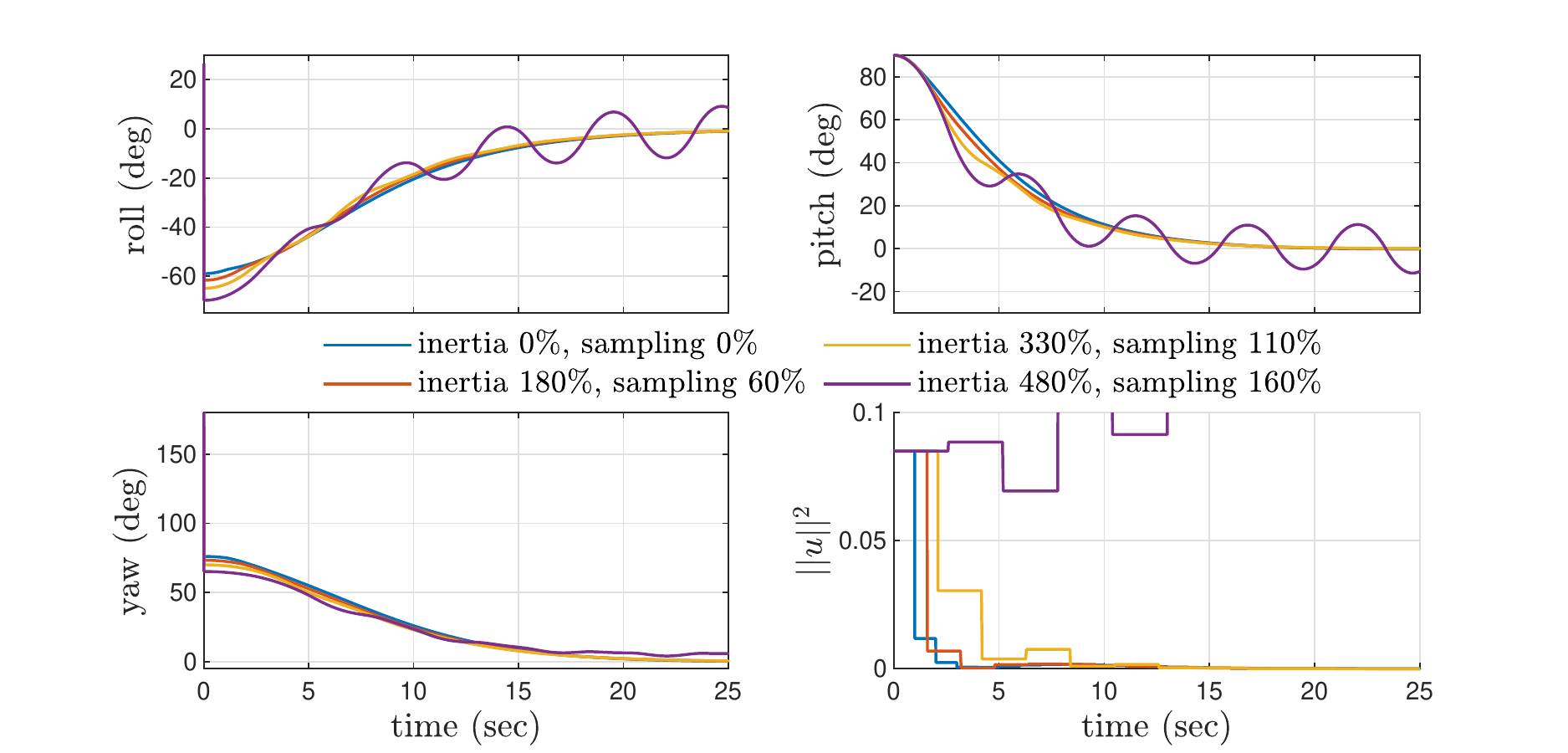}}\vspace{-0.659em}
  
\subfigure[$p=2$]{\includegraphics{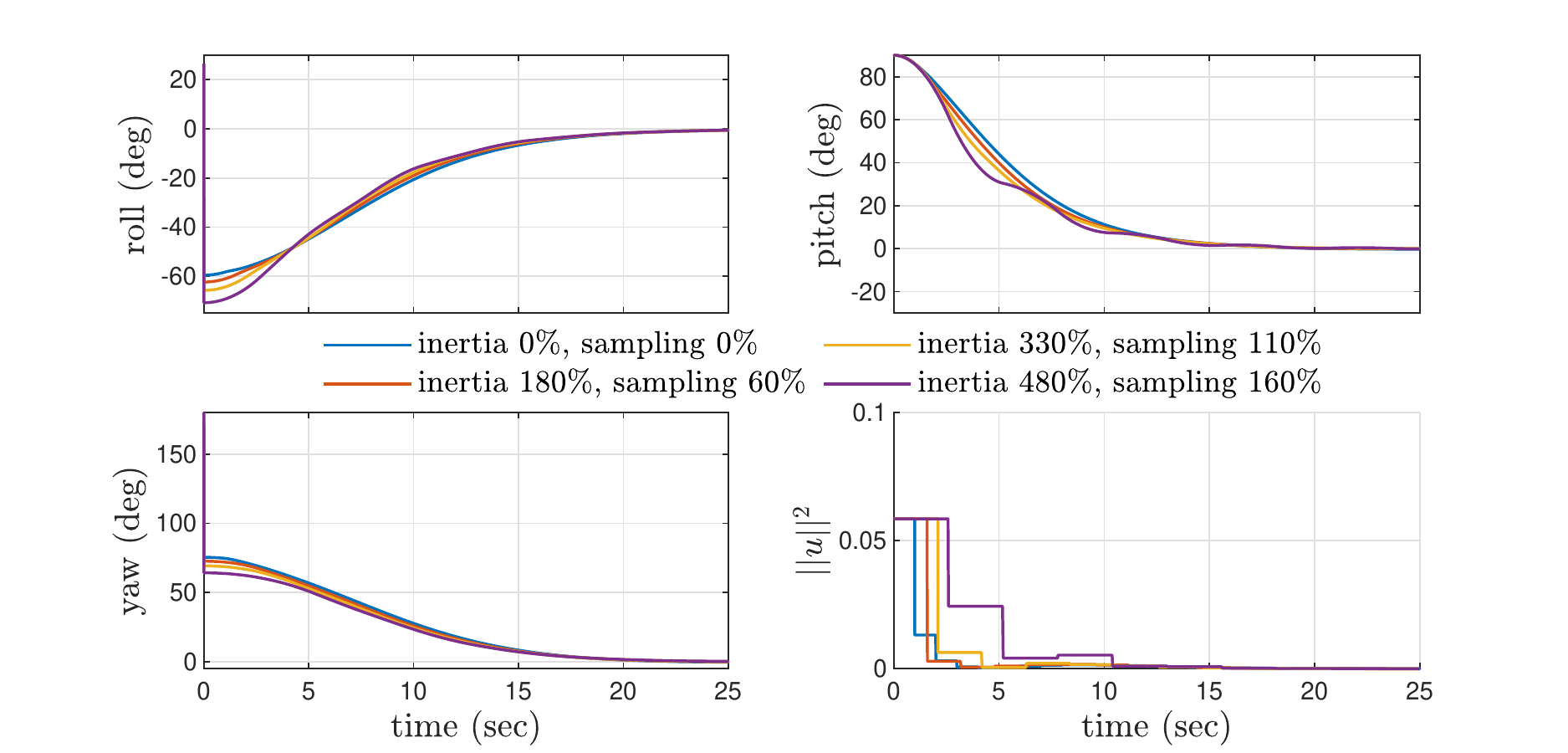}}\hspace{-2em}
  \subfigure[Sampled-data linearized LQR]{\includegraphics{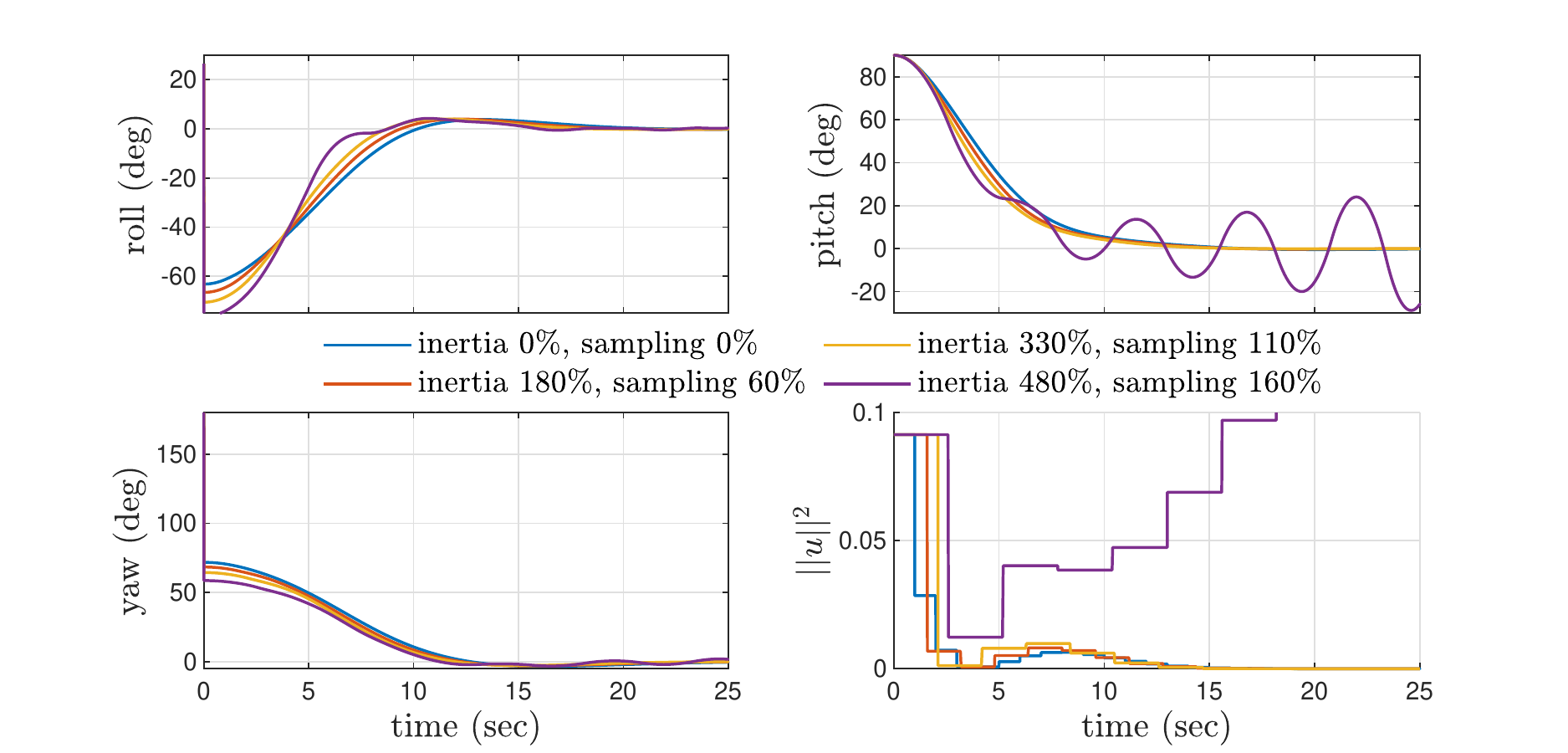}}
  \vspace{-0.5em}
\caption{Stabilization results for nominal values and in presence of uncertainties
}
\label{Figure3}
\end{figure*}

\section{Simulations}\label{sec:sim}
In this section, we consider an attitude maneuver toward $q_0^\star=\mathbf{0}$ of a rigid body having inertia matrix
\begin{align}\label{ex_inertia}
    M = \begin{pmatrix}
    1.42 & 0.00867 & 0.01357 \\ 
    0.00867 & 1.73 & 0.06016  \\ 
    0.01357 & 0.06016 & 2.03
    \end{pmatrix},
\end{align}
assuming initial $\omega(0)=\mathbf{0}$ and roll, pitch, yaw angles $(\phi,\psi,\theta)=(\frac{\pi}{4},\frac{\pi}{2},\pi)$ (corresponding to $q_0(0)=(-0.6533,
    0.2706,
    0.6533
    )^\top$ and $q_r(0)=
    0.2706$), and $\delta=1$\footnote{ Further simulations, with rendering videos, are available at \href{https://youtu.be/u3BhtW-PGPM}{https://youtu.be/u3BhtW-PGPM} where several sampling periods are considered.}.  For the sake of comparison, in the simulations we have considered percentages of uncertainties arising in the elements $M_{22},M_{23},M_{32}$ of the inertia matrix \eqref{ex_inertia} and in the sampling period $\delta$. In particular, simulations in Fig. \ref{Figure3} show the behaviour of the continuous-time controller \eqref{eq:uida_ct}, for $\kappa_\text{di}=\diag(1.1, 0.7, 0.9)$, and compared with the sampled-data controller \eqref{eq:app_mech} for different approximation orders (such as $p=0$ and $p=2$), and the sampled LQR control in \cite{jiang2017sampled}. 
    
    Fig. \ref{Figure3} highlights that although in the nominal case all the controllers achieve $q_0^\star$ with similar performances, as uncertainties occur into the model all the controllers behave differently. In particular, the continuous-time IDA-PBC shows a robust behaviour also with respect to large uncertainties as depicted in Fig. \ref{Figure3}a. Differently, both the sampled-data emulation design, i.e. \eqref{eq:app_mech} with $p=0$ (Fig. \ref{Figure3}b), and the sampled LQR (Fig. \ref{Figure3}d) show a significantly reduced robusteness to parametric uncertainties leading to instability for large uncertainties in both the sampling period and the inertia matrix. Finally, as readily seen in Fig. \ref{Figure3}c, the sensitivity to uncertainty is reduced when considering two correcting terms in the sampled-data emulation design, i.e. \eqref{eq:app_mech} with $p=2$. In fact, the sampled-data controller with $p=2$ shows in all situations comparable performance with the continuous-time one and achieves stabilization even for larger uncertainties.

\section{Conclusions and perspectives} \label{sec:conclusions}

A new quaternion-based digital control for attitude stabilizationhas been proposed. The solution involves, for the first time, discrete-time IDA-PBC over the sampled-data equivalent model so providing, in principle, a simple controller exploiting the physical properties of the dynamics. The feedback gets the form of a series expansion in powers of $\delta$ so that approximations can be naturally and efficiently applied in practice. Future works involve the case of input delays and the extension to formation control of swarms of satellites \cite{wertz1999space}. Also, the extension of IDA-PBC  for general representations of the kinematics is undergoing so to cope with unwinding or singularity phenomena typically rising with parametrizations \cite{bhat2000topological}.

\section*{Acknowledgement}
The Authors wish to thank the Associate Editor and the Reviewers for their valuable comments and suggestions which allowed to improve the quality of the paper.

\bibliographystyle{IEEEtran}      
\bibliography{biblio}                  

\begin{thebibliography}{10}
\providecommand{\url}[1]{#1}
\csname url@samestyle\endcsname
\providecommand{\newblock}{\relax}
\providecommand{\bibinfo}[2]{#2}
\providecommand{\BIBentrySTDinterwordspacing}{\spaceskip=0pt\relax}
\providecommand{\BIBentryALTinterwordstretchfactor}{4}
\providecommand{\BIBentryALTinterwordspacing}{\spaceskip=\fontdimen2\font plus
\BIBentryALTinterwordstretchfactor\fontdimen3\font minus
  \fontdimen4\font\relax}
\providecommand{\BIBforeignlanguage}[2]{{%
\expandafter\ifx\csname l@#1\endcsname\relax
\typeout{** WARNING: IEEEtran.bst: No hyphenation pattern has been}%
\typeout{** loaded for the language `#1'. Using the pattern for}%
\typeout{** the default language instead.}%
\else
\language=\csname l@#1\endcsname
\fi
#2}}
\providecommand{\BIBdecl}{\relax}
\BIBdecl

\bibitem{kim2007robust}
S.-J. Kim, C.-K. Ryoo, and K.~Choi, ``Robust attitude control via quaternion
  feedback linearization,'' in \emph{SICE Annual Conference 2007}.\hskip 1em
  plus 0.5em minus 0.4em\relax IEEE, 2007, pp. 2234--2239.

\bibitem{navabi2017spacecraft}
M.~Navabi and M.~Hosseini, ``Spacecraft quaternion based attitude input-output
  feedback linearization control using reaction wheels,'' in \emph{2017 8th
  International Conference on Recent Advances in Space Technologies
  (RAST)}.\hskip 1em plus 0.5em minus 0.4em\relax IEEE, 2017, pp. 97--103.

\bibitem{mino2019backstepping}
A.~Mino, K.~Uchiyama, and K.~Masuda, ``Backstepping control for satellite
  attitude control using spherical control moment gyro,'' in \emph{2019 SICE
  International Symposium on Control Systems (SICE ISCS)}.\hskip 1em plus 0.5em
  minus 0.4em\relax IEEE, 2019, pp. 39--42.

\bibitem{zhu2011adaptive}
Z.~Zhu, Y.~Xia, and M.~Fu, ``Adaptive sliding mode control for attitude
  stabilization with actuator saturation,'' \emph{IEEE Transactions on
  Industrial Electronics}, vol.~58, no.~10, pp. 4898--4907, 2011.

\bibitem{wen1991attitude}
J.-Y. Wen and K.~Kreutz-Delgado, ``The attitude control problem,'' \emph{IEEE
  Transactions on Automatic control}, vol.~36, no.~10.

\bibitem{hassrizal2016survey}
H.~Hassrizal and J.~Rossiter, ``A survey of control strategies for spacecraft
  attitude and orientation,'' in \emph{2016 UKACC 11th international conference
  on control (CONTROL)}.\hskip 1em plus 0.5em minus 0.4em\relax IEEE, 2016, pp.
  1--6.

\bibitem{lizarralde1996attitude}
F.~Lizarralde and J.~T. Wen, ``Attitude control without angular velocity
  measurement: A passivity approach,'' \emph{IEEE transactions on Automatic
  Control}, vol.~41, no.~3, pp. 468--472, 1996.

\bibitem{tsiotras1998further}
P.~Tsiotras, ``Further passivity results for the attitude control problem,''
  \emph{IEEE Transactions on Automatic Control}, vol.~43, no.~11.

\bibitem{di2003passive}
S.~Di~Gennaro, ``Passive attitude control of flexible spacecraft from
  quaternion measurements,'' \emph{Journal of optimization theory and
  applications}, vol. 116, no.~1, pp. 41--60, 2003.

\bibitem{qasim2017pid}
M.~Qasim, E.~Susanto, and A.~S. Wibowo, ``Pid control for attitude
  stabilization of an unmanned aerial vehicle quad-copter,'' in \emph{2017 5th
  International Conference on Instrumentation, Control, and Automation
  (ICA)}.\hskip 1em plus 0.5em minus 0.4em\relax IEEE, 2017, pp. 109--114.

\bibitem{hughes2012spacecraft}
P.~C. Hughes, \emph{Spacecraft attitude dynamics}.\hskip 1em plus 0.5em minus
  0.4em\relax Wiley, 2012.

\bibitem{bhat2000topological}
S.~P. Bhat and D.~S. Bernstein, ``A topological obstruction to continuous
  global stabilization of rotational motion and the unwinding phenomenon,''
  \emph{Systems \& Control Letters}, vol.~39, no.~1.

\bibitem{chaturvedi2011rigid}
N.~A. Chaturvedi, A.~K. Sanyal, and N.~H. McClamroch, ``Rigid-body attitude
  control,'' \emph{IEEE Control Systems Magazine}, vol.~31, no.~3.

\bibitem{monaco2007advanced}
S.~Monaco and D.~Normand-Cyrot, ``Advanced tools for nonlinear sampled-data
  systems’ analysis and control,'' \emph{European journal of control},
  vol.~13, no. 2-3, pp. 221--241, 2007.

\bibitem{monaco1986linearizing}
S.~Monaco, D.~Normand-Cyrot, and S.~Stornelli, ``On the linearizing feedback in
  nonlinear sampled data control schemes,'' in \emph{1986 25th IEEE Conference
  on Decision and Control}, 1986, pp. 2056--2060.

\bibitem{mattioni2017immersion}
M.~Mattioni, S.~Monaco, and D.~Normand-Cyrot, ``Immersion and invariance
  stabilization of strict-feedback dynamics under sampling,''
  \emph{Automatica}, vol.~76, pp. 78--86, 2017.

\bibitem{jiang2017sampled}
B.~Jiang, Y.~Liu, and K.~I. Kou, ``Sampled-data control for spacecraft attitude
  control systems based on a quaternion model,'' in \emph{2017 Chinese
  Automation Congress (CAC)}.\hskip 1em plus 0.5em minus 0.4em\relax IEEE,
  2017, pp. 4297--4300.

\bibitem{moreschini2020stabilization}
A.~Moreschini, M.~Mattioni, S.~Monaco, and D.~Normand-Cyrot, ``Stabilization of
  discrete port-hamiltonian dynamics via interconnection and damping
  assignment,'' \emph{IEEE Control Systems Letters}, vol.~5, no.~1, pp.
  103--108, 2020.

\bibitem{moreschini2019discrete}
------, ``Discrete port-controlled {H}amiltonian dynamics and average
  passivation,'' in \emph{58th IEEE Conference on Decision and Control (CDC)},
  2019, pp. 1430--1435.

\bibitem{9749871}
S.~Monaco, D.~Normand-Cyrot, M.~Mattioni, and A.~Moreschini, ``Nonlinear
  hamiltonian systems under sampling,'' \emph{IEEE Transactions on Automatic
  Control}, pp. 1--1, 2022.

\bibitem{ortega2002interconnection}
R.~Ortega, A.~Van Der~Schaft, B.~Maschke, and G.~Escobar, ``Interconnection and
  damping assignment passivity-based control of port-controlled {H}amiltonian
  systems,'' \emph{Automatica}, vol.~38, no.~4, pp. 585--596, 2002.

\bibitem{mazenc2016quaternion}
F.~Mazenc, S.~Yang, and M.~R. Akella, ``Quaternion-based stabilization of
  attitude dynamics subject to pointwise delay in the input,'' \emph{Journal of
  Guidance, Control, and Dynamics}, vol.~39, no.~8, 2016.

\bibitem{monaco2011nonlinear}
S.~Monaco and D.~Normand-Cyrot, ``Nonlinear average passivity and stabilizing
  controllers in discrete time,'' \emph{Systems \& Control Letters}, vol.~60,
  no.~6, pp. 431--439, 2011.

\bibitem{mazenc2013robustness}
F.~Mazenc, M.~Malisoff, and T.~N. Dinh, ``Robustness of nonlinear systems with
  respect to delay and sampling of the controls,'' \emph{Automatica}, vol.~49,
  no.~6, pp. 1925--1931, 2013.

\bibitem{wertz1999space}
J.~R. Wertz, W.~J. Larson, D.~Kirkpatrick, and D.~Klungle, \emph{Space mission
  analysis and design}.\hskip 1em plus 0.5em minus 0.4em\relax Springer, 1999,
  vol.~8.

\end{thebibliography}

\end{document}